\documentclass[10pt, conference, letterpaper]{IEEEtran}%
\usepackage{amsmath}
\usepackage{bbm}
\usepackage{graphicx}
\usepackage{epsfig}
\usepackage{array}
\usepackage{psfrag}
\usepackage{amssymb}
\usepackage{mathdots}
\usepackage{color}
\usepackage{pstricks,pst-node,pst-text,pst-3d,pst-plot}
\usepackage{psfrag}
\usepackage{enumerate}
\usepackage{url,cite}
\usepackage{amsfonts}%
\usepackage{amsthm}
\usepackage{dsfont}%
\usepackage{verbatim}%
\usepackage{setspace}
\usepackage{float}
\usepackage{url}
\usepackage{fancyhdr}
\usepackage{bm}
\usepackage{algorithm}
\usepackage{algorithmic}
\newtheorem{thm}{Theorem}
\newtheorem{lma}{Lemma}
\newtheorem{Def}{Definition}
\DeclareMathOperator{\E}{\mathbb{E}}

\newcommand{\lb}{\left (}
\newcommand{\rb}{\right )}

\newcommand{\script}[1]{{\mathcal {#1}}}

\newcommand{\Pmax}{P_{\rm max}}

\newcommand{\Iinst}{I_{\rm inst}}

\newcommand{\fgammai}{f_{\gamma_i}}
\newcommand{\fgi}{f_{g_i}}
\newcommand{\EE}[1]{\E \left[ #1 \right]}
\newcommand{\EEC}[1]{\E_{\bfY(k)} \left[ #1 \right]}
\newcommand{\bgi}{\bar{g}_i}
\newcommand{\bW}{\overline{W}}

\newcommand{\bfP}{{\bf P}}

\newcommand{\bfpi}{{\bm{\pi}}}

\newcommand{\bgamma}{\bar{\gamma}}

\newcommand{\Yivq}{\{Y_i(k)\}_{k=0}^\infty}

\newcommand{\bfY}{{\bf Y}}

\newcommand{\pardef}[1]{\triangleq [#1_1^{(t)},\cdots,#1_N^{(t)}]^T}
\newcommand{\parFdef}[1]{\triangleq [#1_1(k),\cdots,#1_N(k)]^T}
\newcommand{\DOIC}{\emph{DOIC}}

\newcommand{\brho}{\bar{\rho}}
\newcommand{\Ri}{R_i^{(t)}}
\newcommand{\FDurK}{T_k}

\newcommand{\Ts}{T}

\newcommand{\SU}[1]{SU_{#1}}
\newcommand{\Rmax}{R_{\rm max}}
\newcommand{\gammamax}{\gamma_{\rm max}}
\newcommand{\Pit}{P_i^{(t)}}
\newcommand{\git}{g_i^{(t)}}
\newcommand{\HOL}{L^{\rm rem}}
\newcommand{\gmax}{g_{\rm max}}
\newcommand{\dLow}{45}
\newcommand{\dHigh}{60}
\newcommand{\sij}{s_{(j)}}

\begin{document}
\title{Delay-Optimal Scheduling and Power Control for Instantaneous-Interference-Limited CRs}
\author{Ahmed Ewaisha, Cihan Tepedelenlio\u{g}lu\\
\small{School of Electrical, Computer, and Energy Engineering, Arizona State University, USA.}\\
\small{Email:\{ewaisha, cihan\}@asu.edu}\\
}
\maketitle
\begin{abstract}
We study an uplink multi secondary user (SU) cognitive radio system suffering statistical heterogeneity among SUs' channels. This heterogeneity may result in differentiated delay performances to these SUs and result in harmful interference to the PU. We first derive an explicit closed-form expression for the average delay in terms of an arbitrary power-control policy. Then, we propose a delay-optimal closed-form scheduling and power-control policy that can provide the required average delay guarantees to all SUs besides protecting the PU from harmful interference. We support our findings by extensive system simulations and show that it outperforms existing policies substantially.
\end{abstract}

\section{Introduction}

The problem of scarcity in the spectrum band has led to a wide interest in cognitive radio (CR) networks. CRs refer to devices that coexist with the licensed spectrum owners called the primary users (PUs). CR users, also referred to as the secondary users (SUs), located physically close to the PUs might suffer larger degradation in their QoS compared to those that are far because closer SUs transmit with smaller amounts of power. This problem does not appear in conventional non CR cellular systems since frequency channels tend to be orthogonal in non CR systems.

The problem of scheduling and/or power control for CR systems has been widely studied in the literature (please see \cite{Letaief_PU_Known_Location,NEP_Distributed,Ewaisha_TVT2015,Iter_Bit_Allocation_OFDM,6464638}, and references therein). 
The policies proposed in these works aim at optimizing the throughput for the SUs and, at the same time, protecting the PUs from interference. However, providing guarantees on the queuing delay in CR systems was not the goal of these works. In real-time applications, such as audio/video conference calls, packets are expected to arrive at the destination before a prespecified deadline. Thus, the average packet delay needs to be as small as possible to prevent jitter and to guarantee acceptable QoS for these applications \cite{shakkottai2002scheduling,kang2013performance}.

Queuing delay has gained strong attention recently in the literature and scheduling algorithms have been proposed to guarantee small delay \cite{neely2003power,Two_Q_Light_Hvy,li2011delay}. A power control and routing policy is proposed in \cite{neely2003power} to maximize the capacity region under an instantaneous power constraint. While the authors show an upper bound on the average delay, this delay performance is not guaranteed to be optimal. Reference \cite{li2011delay}, which is the closest to our work, studies the joint scheduling-and-power-control problem. It assumes that all users process packets with the same power since it discusses the problem of processing jobs at a CPU. The CPU problem is a special case of the wireless channel problem herein where the CPU processes jobs at a deterministic rate. The authors assume priority scheduling and depend on the closed-form expressions for the average delay. Up to our knowledge, closed-form expressions for the average packet delay do not exist in the literature for the random rate case.



The authors of \cite{ewaisha2015dynamic} propose a scheduling policy to minimize the sum of SUs' average delays. However, power control was out of their scope. In CRs, power control dictates adhering to PU's, instantaneous or average, interference constraints. In this paper, we extend the work in \cite{ewaisha2015dynamic} to study the problem under instantaneous interference constraints. For the average interference case, the reader is referred to \cite{Ewaisha_Globecom16} where we show that the solution to the problem has a totally different structure than the instantaneous case. The contributions in this paper are: i) The derivation of a closed-form expression for the average delay as a function of the power control policy, ii) finding a closed-form expression for the optimal power control policy, iii) proposing a delay-optimal joint power-control-and-scheduling policy that includes a closed-form expression for the power control policy in a CR system. 



The rest of the paper is organized as follows. The system model and the underlying assumptions are presented in Section \ref{Model}. In Section \ref{Prob_Statement} we formulate the problem mathematically. The proposed policy and its optimality are presented in Section \ref{Proposed_Algorithm},
followed by the extensive simulation results in Section \ref{Results}. Finally the paper is concluded in Section \ref{Conclusion}.

\section{System Model}
\label{Model}
We assume a CR system consisting of a single secondary base station (BS) serving $N$ secondary users (SUs) indexed by the set $\script{N}\triangleq \{1,\cdots N\}$ (Fig. \ref{Cell_Fig}). We are considering the uplink phase where each SU has its own buffer for packets that need to be sent to the BS. The SUs share a single frequency channel with a single PU that has licensed access to this channel. The CR system operates in an underlay fashion where the PU is using the channel continuously at all times. SUs are allowed to transmit as long as they do not cause harmful interference to the PU. We assume an instantaneous interference constraint where the interference received by the PU at any given slot should not exceed a prespecified threshold $\Iinst$.

\begin{figure}%
\centering
\includegraphics[width=0.7\columnwidth]{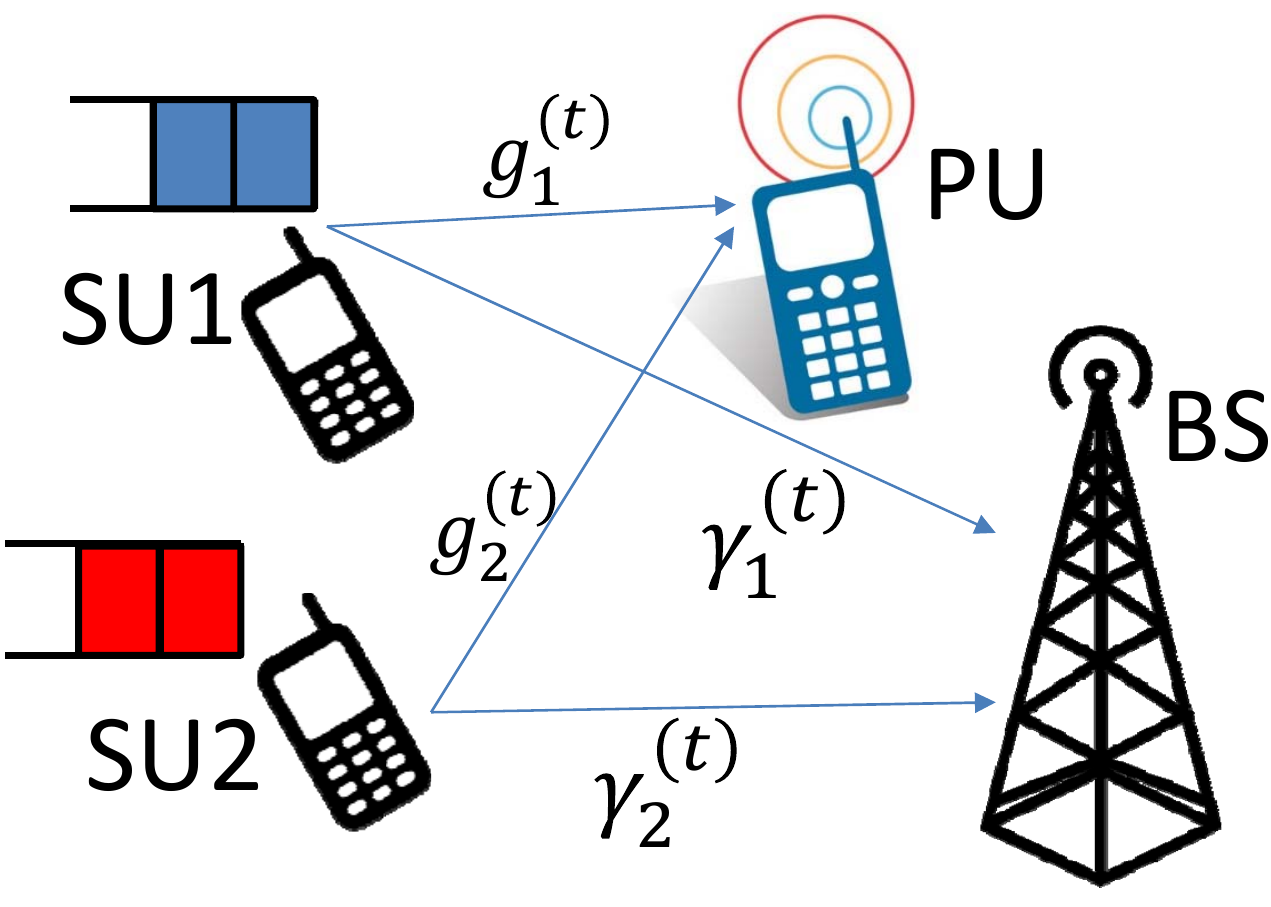}%
\caption{An uplink CR system with $N$ SUs ($N=2$ in this figure) communicating with their BS. There exists an interference link between each SU and the existing PU. The PU is assumed to be using the channel continuously.}%
\label{Cell_Fig}%
\end{figure}

\subsection{Channel and Interference Model}
We assume a time slotted structure where each slot is of duration $\Ts$ seconds. The channel between $\SU{i}$ and the BS (and that between $\SU{i}$ and the PU) is block fading with instantaneous power gain $\gamma_i^{(t)}$ (with gain $g_i^{(t)}$), at time slot $t$, following the probability mass function $\fgammai(\gamma)$ with mean $\bgamma_i$ ($\fgi(g)$ with mean $\bgi$) and i.i.d. across time slots, and $\gammamax$ ($\gmax$) is the maximum gain $\gamma_i^{(t)}$ ($\git$) could take. The channel gains are statistically independent and heterogeneous across SUs. We assume perfect knowledge of $\gamma_i^{(t)}$ and $g_i^{(t)}$ at the beginning of slot $t$ through some channel estimation phase that is out of the scope of this work (see \cite[Section VI]{haykin2005cognitive} and \cite{Bari13ciss,Bari13asilomar,Bari14asilomar,Bari15spl,Bari15ciss1,Bari15asilomar1} for different channel estimation techniques in CRs). At time slot $t$, $\SU{i}$ transmits with rate $\Ri(\Pit)\triangleq\log \lb 1+P_i^{(t)}\gamma_i^{(t)} \rb$, where $P_i^{(t)}\leq\Pmax$ is is $\SU{i}$'s transmission power at slot $t$, for some maximum value $\Pmax$. We assume that there exists a finite maximum rate $\Rmax\triangleq\log \lb 1+\Pmax\gammamax \rb$ that $\SU{i}$ cannot exceed.

\subsection{Queuing Model}
\subsubsection{Arrival Process} We assume that packets arrive to the $\SU{i}$'s buffer at the beginning of each slot. The number of packets arriving to $\SU{i}$'s buffer follows a Bernoulli process with a fixed parameter $\lambda_i$ packets per time slot. Following the literature, packets are buffered in infinite-sized buffers \cite[pp. 163]{Bertsekas_Data_Networks} and are served according to the first-come-first-serve discipline. Each packet has a fixed length of $L$ bits that is constant for all users where $L\gg\Rmax$ which is a typical case for packets with large sizes as video packets \cite{semiconductor2008long}.

\subsubsection{Service Process} When $\SU{i}$ is scheduled for transmission at slot $t$, it transmits $M_i^{(t)}$ bits of the head-of-line (HOL) packet of its queue. The remaining bits of this HOL packet remain in the HOL of $\SU{i}$'s queue until it is reassigned the channel in subsequent time slots. Here $M_i^{(t)}\triangleq\min \lb \Ri,\HOL_i(t) \rb$ bits, where $\HOL_i(t+1)\triangleq \HOL_i(t) - M_i^{(t)}$ is the remaining number of bits of the HOL packet at $\SU{i}$ at the end of slot $t$. $\HOL_i(t)$ is initialized by $L$ whenever a packet joins the HOL position of $\SU{i}$'s queue so that it always satisfies $0\leq \HOL_i(t)\leq L$, $\forall t$. A packet is not considered transmitted unless all its $L$ bits are transmitted, i.e. unless $\HOL_i(t)$ becomes zero, at which point $\SU{i}$'s queue decreases by 1 packet. At the beginning of slot $t+1$ the following packet in the buffer, if any, becomes $\SU{i}$'s HOL packet and $\HOL_i(t+1)$ is reset back to $L$ bits. The $\SU{i}$'s queue evolves as $Q_i^{(t+1)}= (Q_i^{(t)} + \vert \script{A}_i^{(t)}\vert - S_i^{(t)})^+$, where $\script{A}_i^{(t)}$ is the set carrying the index of the packet, if any, arriving to $\SU{i}$ at slot $t$, thus $\vert \script{A}_i^{(t)}\vert$ is either $0$ or $1$ since at most one packet per slot can arrive to $\SU{i}$; 
the packet service indicator $S_i^{(t)}=1$ if $\HOL_i(t)$ becomes zero at slot $t$.

The service time $s_i(P_i)$ of $\SU{i}$ is the number of time slots required to transmit one packet for $\SU{i}$, excluding the service interruptions, for some arbitrary power control policy $P_i$. The service time is assumed to follow a general distribution throughout the paper that depends on the distribution of $\Pit\gamma_i^{(t)}$.

We define the delay $W_i^{(j)}$ of a packet $j$ as the total number of time slots packet $j$ spends in $\SU{i}$'s buffer from the slot it joined the queue until the slot when its last bit is transmitted. 
The time-average delay experienced by $\SU{i}$'s packets is given by \cite{li2011delay}
\begin{equation}
\bW_i \triangleq \lim_{T \rightarrow \infty} \frac{\EE{\sum_{t=1}^T{\sum_{j\in\script{A}_i^{(t)}}{W_i^{(j)}}}}}{\EE{\sum_{t=1}^T{\vert \script{A}_i^{(t)}\vert}}}
\label{Delay}
\end{equation}
which is the expected total amount of time spent by all packets arriving in a time interval, of a large duration, normalized by the expected number of packets that arrived in this interval.


\section{Problem Statement}
\label{Prob_Statement}
Each $\SU{i}$ has an average delay constraint $\bW_i \leq d_i$ that needs to be satisfied. Moreover, the PU can tolerate an interference level of $\Iinst$ at any given slot. In this work, we are interested in frame-based scheduling policies where frame $k$ consists of a random number $T_k\triangleq\vert \script{F}(k)\vert$ of consecutive time-slots, where $\script{F}(k)$ is the set containing the indices of the time slots belonging to frame $k$ (see Fig. \ref{Frame_Structure}). The idea of dividing time into frames and assigning fixed scheduling policy for each frame was also used in \cite{li2011delay}. Where each frame begins and ends is specified by idle periods and will be precisely defined later in this section. During frame $k$, SUs are scheduled according to some preemptive-resume \cite[pp. 205]{Bertsekas_Data_Networks} priority list $\bfpi(k)\parFdef{\pi}$ that is fixed during the entire frame $k$, where $\pi_j(k)$ is the index of the SU who is given the $j$th priority during frame $k$. Each frame consists of exactly one \emph{idle period} followed by exactly one \emph{busy period}, both are defined next.

\begin{figure}
\centering
\includegraphics[width=0.8\columnwidth]{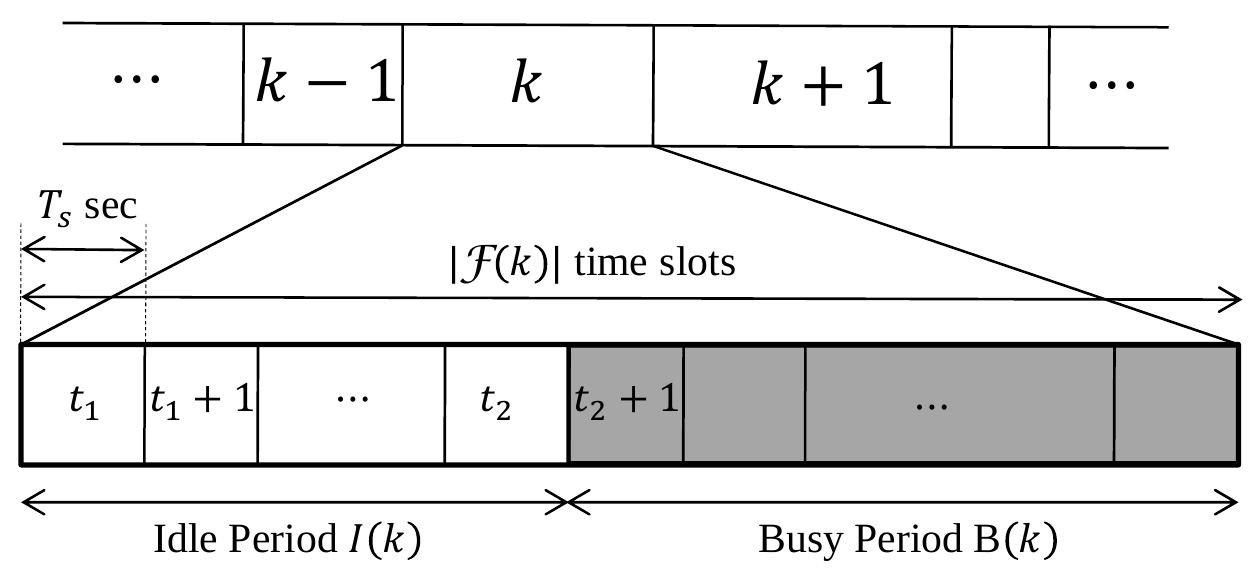}%
\caption{Time is divided into frames. Frame $k$ has $\FDurK\triangleq\vert \script{F}(k)\vert$ slots, each is of duration $\Ts$ seconds. Different frames can have different number of time slots.}%
\label{Frame_Structure}%
\end{figure}

\begin{Def}
\label{Idle_Def}
An idle period is the time interval formed by the consecutive time slots where all SUs have empty buffers. An idle period starts with the time slot $t_1$ following the completion of transmission of the last packet in the system, and ends with a time slot $t_2$ when one or more of the SUs' buffer receives one a new packet to be transmitted (see Fig. \ref{Frame_Structure}).
\end{Def}

\begin{Def}
\label{Busy_Def}
Busy period is the time interval between two consecutive idle periods.
\end{Def}
The duration of the idle period $I(k)$ and busy period $B(k)$ of frame $k$ are random variables, thus $T_k=I(k)+B(k)$ is random as well. Since frames do not overlap, if $t\in\script{F}(k_1)$ then $t \notin \script{F}(k_2)$ as long as $k_1\neq k_2$. Given some priority list $\bfpi$, after dropping the index $k$ for simplicity, $\SU{\pi_j}$'s statistical average delay $\bW_{\pi_j}^{\rm s}(k)$ is \cite[pp. 205]{Bertsekas_Data_Networks}
\begin{equation}
\frac{\EE{s_{\pi_j}(P_{\pi_j})}}{\lb 1-\brho_{\pi_{j-1}}\rb}+ \frac{\sum_{l=1}^j\lambda_{\pi_l}\EE{s_{\pi_l}^2(P_{\pi_l})}}{\lb 1-\brho_{\pi_{j-1}}\rb\lb 1-\brho_{\pi_{j-1}} - \rho_{\pi_j}(P_{\pi_j})\rb}.
\label{Priority_Delay}
\end{equation}
We note that $\bW_{\pi_j}^{\rm s}$ depends on $\EE{s_{\pi_l}(P_{\pi_l})}$ and $\EE{s_{\pi_l}^2(P_{\pi_l})}$ that are not known in closed-form in the literature.

The main objective of this work is to solve
\begin{equation}
\begin{array}{ll}
\underset{\{\bfP^{(t)}\},\{\bfpi(k)\}}{\rm{minimize}}& \sum_{i=1}^N \bW_i\\
\label{Problem}
\rm{subject \; to} & \sum_{i=1}^N {P_i^{(t)}g_i^{(t)}} \leq \Iinst \hspace{0.25in} , \hspace{0.25in} \forall t\geq 1\\
& \bW_i \leq d_i\\
& P_i^{(t)} \leq \Pmax \hspace{0.1in} , \hspace{0.1in} \forall i \in \script{N} \rm{\; and \;}\forall t \geq 1,\\
& \sum_{i=1}^N{  \mathds{1} \lb P_i^{(t)}\rb} \leq 1 \hspace{0.25in}, \hspace{0.25in} \forall t \geq 1,
\end{array}
\end{equation}
where $\bfP^{(t)} \pardef{P}$, while $\mathds{1}(x)\triangleq1$ if $x\neq 0$ and $0$ otherwise. The last constraint indicates that no more than a single SU is to be transmitting at slot $t$. 
We next propose a low complexity update policy and show its optimality.

\section{Proposed Power Allocation and Scheduling Policy}
\label{Proposed_Algorithm}
We solve the problem by proposing an online policy that dynamically updates the power control vector $\bfP^{(t)}$ and the priority vector $\bfpi(k)$. We show that this policy has an asymptotically optimal performance. That is, we can achieve a delay arbitrarily close to the optimal value depending on some control parameter $V$.

\subsection{Satisfying Delay constraints}
In order to guarantee a feasible solution satisfying the delay constraints in problem \eqref{Problem}, we set up a ``virtual queue'' associated with each delay constraint in problem \eqref{Problem}. The virtual queue for $\SU{i}$ at frame $k$ is given by
\begin{equation}
Y_i(k+1)=\lb Y_i(k)+\sum_{j\in \script{A}_i(k)}{\lb W_i(j)-r_i(k)\rb} \rb^+
\label{Delay_Q}
\end{equation}
where $r_i(k)\in[0,d_i]$ is an auxiliary variable, that is to be optimized over, while $\script{A}_i(k)\triangleq\cup_{t\in\script{F}(k)}\script{A}_i^{(t)}$ is the set of all packets arrived at $\SU{i}$'s buffer during frame $k$. We define $\bfY(k) \parFdef{Y}$ for notational convenience. Equation \eqref{Delay_Q} is calculated at the end of frame $k$ and represents the amount of delay exceeding the delay bound $d_i$ up to the end of frame $k$. We first mention the following definition, then state a lemma that gives a sufficient condition for the delay of $\SU{i}$ to satisfy $\bW_i \leq d_i$.
\begin{Def}
\label{Mean_Rate_Def}
We say that the random sequence $\{Y_i(k)\}_{k=0}^\infty$ is mean rate stable if and only if $\lim_{K\rightarrow\infty}\EE{Y_i(K)}/K=0$.
\end{Def}

\begin{lma}
\label{Mean_Rate_Lemma}
If $\{Y_i(k)\}_{k=0}^\infty$ is mean rate stable, then the time-average delay of $\SU{i}$ satisfies $\bW_i \leq d_i$.
\end{lma}
\begin{proof}
Lemma 3 of \cite{li2011delay} can be modified to prove that
\begin{equation}
\bW_i \leq \lim_{K \rightarrow \infty} \frac{\EE{\sum_{k=1}^K{\vert \script{A}_i(k) \vert {r_i(k)}}}}{\EE{\sum_{k=1}^T{\vert\script{A}_i(k)\vert}}}.
\label{Wait_r_i}
\end{equation}
The proof follows by replacing $r_i(k)$ by its bound $d_i$ in \eqref{Wait_r_i}.
\end{proof}
Lemma \ref{Mean_Rate_Lemma} says that if the power control and scheduling policy results in a mean rate stable $\Yivq$, then the average delay constraint of problem \eqref{Problem} is satisfied.

\subsection{Proposed Policy}
\label{Algorithm_Subsection}
We now propose the \emph{Delay Optimal with Instantaneous Interference Constraint} (\emph{DOIC}) policy executed at the beginning of each frame $k$ for finding $\bfP^{(t)}$ as well as the optimum list $\bfpi(k)$, given some prespecified control parameter $V$. We first define the following power control policy
\begin{equation}
	P_i^{*(t)}=\min\lb\frac{\Iinst}{g_i^{(t)}},\Pmax\rb,
\label{Power_Allocation}
\end{equation}
that will later be shown to be the optimal power control policy.

\noindent{\bf {\DOIC} Policy:}
\begin{enumerate}
	\item The BS sorts the SUs according to the descending order of $Y_i(k)/\EE{s_i(P_i^*)}$. The sorted list is denoted by $\bfpi(k)$.
	\item At each slot $t\in\script{F}(k)$, among all SUs having non-empty buffers, $\SU{i^*}$, that has the highest priority in $\bfpi(k)$, is assigned the channel and transmits with power $P_{i^*}^{*(t)}$.
	\item At the end of frame $k$, for all $i\in \script{N}$ the BS updates: $r_i(k)= d_i$ if $V<Y_i(k)\lambda_i$, and $r_i(k)=0$ otherwise, and then $Y_i(k+1)$ via \eqref{Delay_Q}, $\forall i\in \script{N}$.
\end{enumerate}

The intuition behind the {\DOIC} policy is as follows. We define the Lyapunov function and Lyapunov drift to be $L \lb \bfY(k)\rb \triangleq \frac{1}{2}\sum_{i=1}^N Y_i^2(k)$ and $\Delta \lb \bfY(k)\rb \triangleq \EEC{L \lb \bfY(k+1) \rb - L \lb \bfY(k) \rb}$, respectively, where $\EEC{x}$ denotes the conditional expectation of $x$ given $\bfY(k)$, namely $\EEC{x} \triangleq \EE{x \vert \bfY(k)}$. Squaring equation \eqref{Delay_Q} and taking the conditional expectation, we get
\begin{equation}
\hspace{-0.3in}\begin{array}{ll}
	&\frac{1}{2}\EEC{Y_i^2(k+1)-Y_i^2(k)} =Y_i(k)\EEC{\FDurK}\lambda_i\times\\
& \hspace{0.35in}\lb\bW_i^{\rm s}(k) - r_i(k)\rb+ \EEC{\lb \sum_{j\in \script{F}_k} \lb W_i^{(j)} - r_i(k)\rb\rb^2}
\end{array}
\label{Delay_Q_Sq1}
\end{equation}
\begin{equation}
\leq Y_i(k) \EEC{\FDurK}\lambda_i \lb \bW_i^{\rm s}(k)-r_i(k)\rb + C_{Y_i},
\label{Delay_Q_Sq2}
\end{equation}
where the inequality in \eqref{Delay_Q_Sq2} comes from upper-bounding the last term in \eqref{Delay_Q_Sq1} by some finite constant $C_{Y_i}<\infty$. The existence of this finite constant can be shown by following steps similar to the proof of Lemma 4 in \cite{li2011delay}. We omit these steps for brevity. Given some fixed control parameter $V>0$, we add the penalty term $V\sum_i \EEC{r_i(k)\FDurK}$ to both sides of \eqref{Delay_Q_Sq2} and rearranging, drift-plus-penalty term becomes bounded by
$\Delta\lb \bfY(k)\rb + V\sum_{i=1}^N \EEC{r_i(k)\FDurK}\leq \EEC{\FDurK}\phi(k)+C_Y$ with $\phi(k)\triangleq\sum_{i=1}^N \lb V-Y_i(k) \lambda_i\rb r_i(k) + \sum_{j=1}^N Y_{\pi_j(k)}(k) \lambda_{\pi_j(k)}\bW_{\pi_j(k)}(k)$.

The {\DOIC} policy is defined as the policy that chooses the values of $\{\Pit\}_{t\in \script{F}(k)}$, $\forall i\in \script{N}$, and $\bfpi(k)$ along with the auxiliary variable $r_i(k)$, $\forall i\in \script{N}$, that minimize $\phi(k)$. Since the two summations in $\phi(k)$ are decoupled, we can use step 4.a to minimize the first summation, then state the next theorem to discuss the optimum power control policy and then finally state the scheduling rule. First, given some priority list $\bfpi$ used in frame $k$, we define $\rho(P_{\pi_j})\triangleq\lambda_{\pi_j}\EE{s_{\pi_j}(P_{\pi_j})}$ and $\brho_{\pi_j}\triangleq \sum_{l=1}^j \rho_{\pi_l}(P_{\pi_l})$, where $P_{\pi_j}$ is any arbitrary power control policy that controls $\SU{\pi_j}$'s power $P_{\pi_j}^{(t)}$ $\forall t\in\script{F}(k)$.
\begin{thm}
\label{Service_Time_Second_Moment}
Given some priority list $\bfpi$ used during frame $k$, the first and second moments of $s_i(P_i)$ in \eqref{Priority_Delay} are given by $\EE{s_i(P_i)}=L/\EE{\Ri(\Pit)}$ and $\EE{s_i^2(P_i)}=\sum_{\tau_1=1}^L\sum_{\tau_2=1}^L\Pr[\sum_{t=1}^{\max(\tau_1,\tau_2)-1}\log(1+\Pit\gamma_i^{(t)})\leq L-1]$, respectively, for any arbitrary power control policy $P_i$. Moreover, the power control policy in \eqref{Power_Allocation} minimizes $\phi(k)$ for any priority list $\bfpi$.
\end{thm}
\begin{proof}
We derive here the $\EE{s_i^2(P_i)}$ while $\EE{s_i(P_i)}$ is derived similarly. Let $\sij$ be the service time of packet $j$. Thus $\sij^2=\lb\sum_{\tau=1}^L \mathds{1}(\tau,j)\rb^2=\sum_{\tau_1=1}^L \sum_{\tau_2=1}^L \mathds{1}(\tau_1,j)\mathds{1}(\tau_2,j)$ where $\mathds{1}(\tau,j)=1$ if any portion of packet $j$ was transmitted at slot $\tau$, and $0$ otherwise, that is
\begin{equation}
\mathds{1}(\tau_1,j)=\left\{
\begin{array}{ll}
	1 &\sum_{t=1}^{\tau_1-1}\log\lb 1+\Pit\gamma_i^{(t)}\rb\leq L-1\\
	0 &\mbox{o.w.}
\end{array}
\right.
\label{Indicator}
\end{equation}
which means that the product $\mathds{1}(\tau_1,j)\mathds{1}(\tau_2,j)$ can be given by
\begin{equation}
\left\{
\begin{array}{ll}
	1 &\sum_{t=1}^{\max\lb\tau_1,\tau_2\rb-1}\log\lb 1+\Pit\gamma_i^{(t)}\rb\leq L-1\\
	0 &\mbox{o.w.}
\end{array}
\right.
\label{Indicator2}
\end{equation}
The time average $\EE{s_i^2(P_i)}$ is given by
\begin{equation}
\lim_{J\rightarrow\infty}\frac{\sum_{j=1}^J \sij^2}{J}=\sum_{\tau_1=1}^L \sum_{\tau_2=1}^L\EE{\mathds{1}(\tau_1,j)\mathds{1}(\tau_2,j)}.
\label{Time_Avg_2nd_Moment}
\end{equation}
Substituting by the expectation of \eqref{Indicator2} in \eqref{Time_Avg_2nd_Moment} completes the first part of the theorem. The second part is derived by showing that both $\EE{s_i(P_i)}$ and $\EE{s_i^2(P_i)}$, and hence $\bW_{\pi_j}^{\rm s}(k)$, are decreasing in $\Pit$ then using the Lagrange optimization to find $\Pit$, $\forall i\in\script{N}$, that minimize $\phi(k)$ to yield \eqref{Power_Allocation}.
\end{proof}

Theorem \ref{Service_Time_Second_Moment} shows the optimum power control policy to minimize $\phi(k)$. Now, to choose $\bfpi(k)$ we use the $c\mu$ rule that minimizes the second summation in $\phi(k)$ as demonstrated in step 1 of the {\DOIC} policy.


\begin{thm}
\label{Optimality}
If \eqref{Problem} is strictly feasible, then there exists some finite constant $C_Y$ such that satisfying
\begin{equation}
\sum_{i=1}^N{\bW_i} \leq \frac{C_Y}{V} + \sum_{i=1}^N{\bW_i^*}
\label{Optimality_Equation}
\end{equation}
when the BS adopts the {\DOIC} policy, where $\bW_i^*$ is the optimum value of the delay when solving problem \eqref{Problem}. Moreover, the virtual queues $\Yivq$ are mean rate stable $\forall i \in \script{N}$.
\end{thm}

\begin{proof}
\label{Optimality_Proof}
It follows along the lines of the proof of Theorem 2 in \cite{li2011delay} and thus omitted due to lack of space.
\end{proof}

Theorem \ref{Optimality} says that the objective function of problem \eqref{Problem} is upper bounded by the sum of the optimum values $\bW_i^*$ plus some constant that vanishes as $V\rightarrow\infty$. The drawback of setting $V$ very large is that the policy converges slower. That is, the virtual queues become mean rate stable after a larger number of frames. Having a vanishing gap means that {\DOIC} is asymptotically optimal. Moreover, since $\Yivq$ is mean rate stable, the constraint $\bW_i\leq d_i$ is satisfied $\forall i$.

\section{Simulation Results}
\label{Results}
We simulated a system of $N=5$ SUs. Unless otherwise specified, Table \ref{Parameters} lists all parameter values. $\SU{i}$'s arrival rate is set to $\lambda_i=i\lambda$ for some fixed parameter $\lambda$.  All SUs are having homogeneous channel conditions except $\SU{5}$ who has the highest $\overline{g}_5$. Thus $\SU{5}$ is statistically the worst case user.
\begin{table}
	\centering
		\caption{Simulation Parameter Values}
		\label{Parameters}
		\begin{tabular}{|c|c||c|c|}
			\cline{1-4}
			Parameter & Value & Parameter & Value \\
			\cline{1-4}
			$d_i$ $\forall i\leq 4$ & $\dHigh\Ts$ &$\gmax$ & $10\bgi$\\
			$d_5$ & $45\Ts$ & $\Iinst$ & 20 \\
			$L$ & $10^3$ bits/packet & $\Pmax$ & 100 \\
			$V$ & $100$ &$\bgamma_i$ $\forall i\leq 5$ & $1$  \\
			$\fgammai(\gamma)$ & $\exp{\lb-\gamma/\bgamma_i\rb}/\bgamma_i$ & $\overline{g}_i$ $\forall i\leq 4$ & $0.1$\\
			$\fgi(g)$ & $\exp{\lb-g/\overline{g}_i\rb}/\overline{g}_i$ & $\overline{g}_5$ & $0.4$
			\\ \cline{1-4}
			\end{tabular}
\end{table}

Fig. \ref{PerSU_Delay_DOIC} plots average per-SU delay $\bW_i$, from \eqref{Delay}, versus $\lambda$. The plot is for the {\DOIC} policy for two cases; the first is with $d_5=\dLow\Ts$ while the second is with $d_5=\dHigh\Ts$. We can see that $\SU{5}$ has the worst average delay. However, the {\DOIC} policy has forced $\bW_5$ to be smaller than $\dHigh\Ts$ for all $\lambda$ values. This comes at the cost of another user's delay. We conclude that the proposed policy can force the delay vector of the SUs to take any value as long as it is strictly feasible.

Fig. \ref{Sum_Delay_DOIC_CSMA_CNC} compares the delay performance of the {\DOIC} to two different schemes, namely the Carrier Sense Multiple Access (CSMA) algorithm, which allocates the channel equally likely among all SUs, as well as the Cognitive Network Control (CNC) algorithm proposed in \cite{Neely_CNC_2009} which is a version of the MaxWeight algorithm. Both schemes allocate the power according to \eqref{Power_Allocation}. The gap differences between these two policies and the {\DOIC} policy are over $8\%$ and $5\%$, respectively, at light traffic, and $16\%$ and $8\%$ at heavy traffic. The reason why the proposed policy outperforms both policies is because the it gives priority to SUs with the worst delay history, while the CSMA and CNC schedule the SUs to guarantee fairness and to maximize the achievable rate region, respectively.


\begin{figure}%
\centering
\includegraphics[width=\columnwidth]{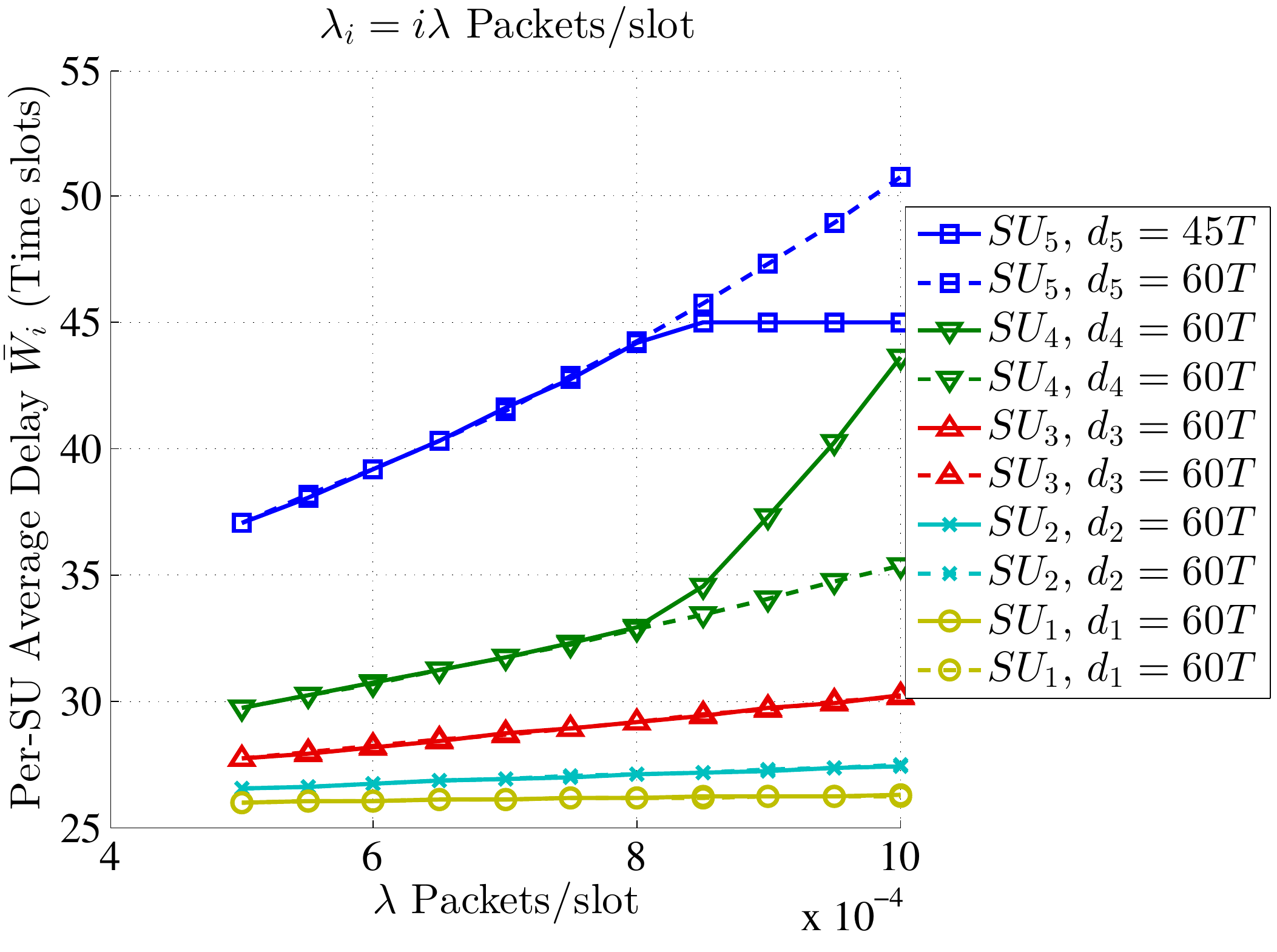}%
\caption{The {\DOIC} policy can force the average delay for any SU to take any value as long as it is strictly feasible.}%
\label{PerSU_Delay_DOIC}%
\end{figure}

\begin{figure}%
\centering
\includegraphics[width=\columnwidth]{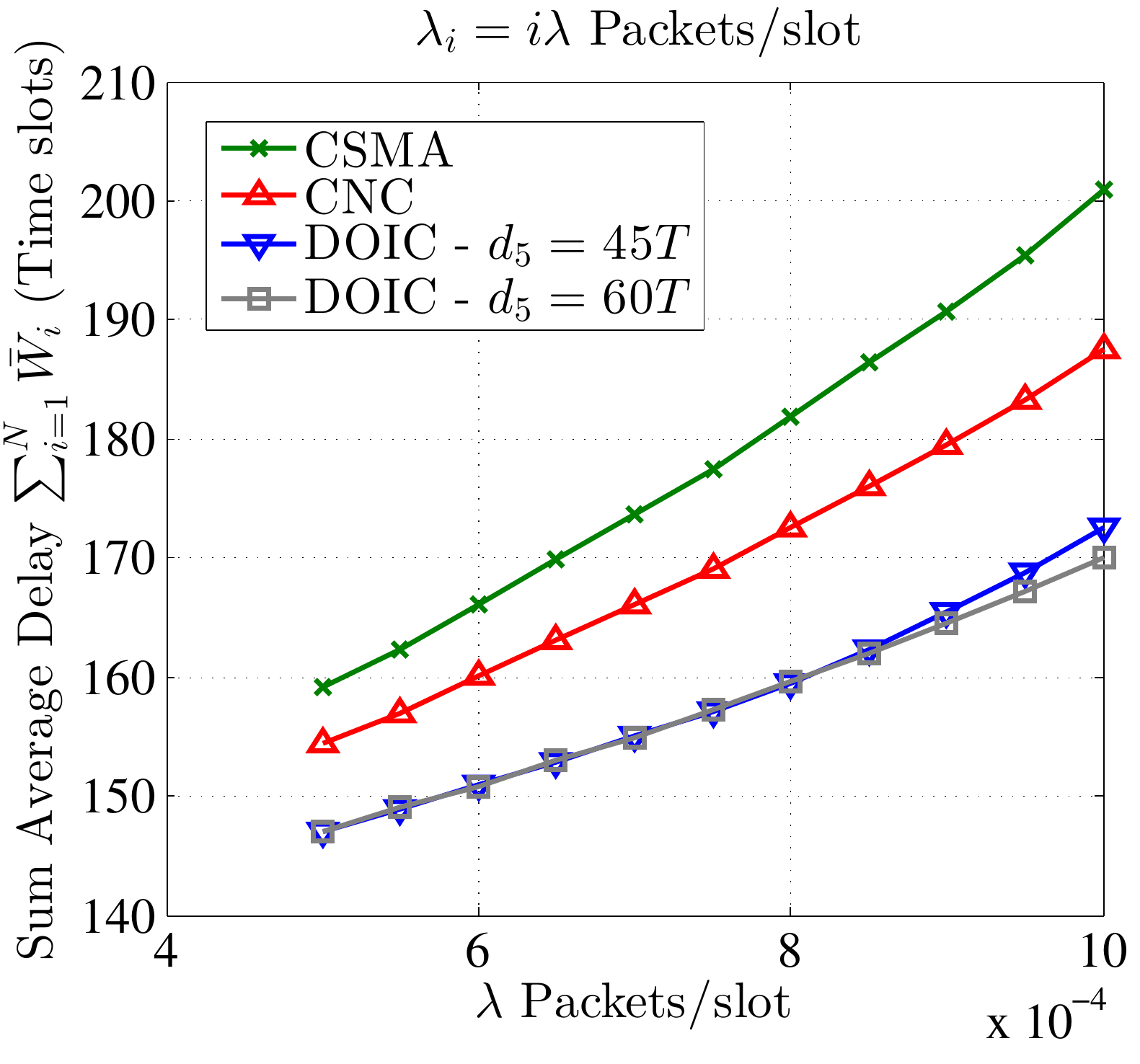}%
\caption{Sum of cost functions for the perfect as well as the imperfect channel sensing for both the constrained and unconstrained optimization problems.}%
\label{Sum_Delay_DOIC_CSMA_CNC}%
\end{figure}

\section{Conclusion}
\label{Conclusion}
We have studied the joint scheduling and power control problem of an uplink multi SU CR system. We formulated the problem as a delay minimization problem in the presence of instantaneous interference constraints to the PU. Most of the existing literature that study this problem either assume on-off fading channels or do not provide a delay-optimal policies which is essential for real-time applications. We derived closed-form expressions for the average delay in terms of any arbitrary power policy which helped in proposing a closed-form power-control policy as well as a low-complexity scheduling policy. We showed, through the Lyapunov optimization, that the proposed policy is asymptotically delay optimal. That is, it minimizes the sum of average delays of the SUs as well as satisfying the instantaneous interference and average delay constraints. Extensive simulation results showed that the proposed policy outperforms existing policies.





\bibliographystyle{IEEEbib}
\bibliography{MyLib}

\begin{thebibliography}{10}

\bibitem{Letaief_PU_Known_Location}
K.~Hamdi, Wei Zhang, and K.B. Letaief,
\newblock ``Uplink scheduling with qos provisioning for cognitive radio
  systems,''
\newblock in {\em Wireless Communications and Networking Conference, 2007.WCNC
  2007. IEEE}, march 2007, pp. 2592 --2596.

\bibitem{NEP_Distributed}
Z.~Guan, T.~Melodia, and G.~Scutari,
\newblock ``To transmit or not to transmit? distributed queueing games in
  infrastructureless wireless networks,''
\newblock {\em Networking, IEEE/ACM Transactions on}, vol. PP, no. 99, pp.
  1--14, 2015.

\bibitem{Ewaisha_TVT2015}
A.E. Ewaisha and C.~Tepedelenlio\u{g}lu,
\newblock ``Throughput optimization in multichannel cognitive radios with
  hard-deadline constraints,''
\newblock {\em IEEE Transactions on Vehicular Technology}, vol. 65, no. 4, pp.
  2355--2368, April 2016.

\bibitem{Iter_Bit_Allocation_OFDM}
Yonghong Zhang and C.~Leung,
\newblock ``Resource allocation in an {OFDM}-based cognitive radio system,''
\newblock {\em IEEE Transactions on Communications}, vol. 57, no. 7, pp.
  1928--1931, July 2009.

\bibitem{6464638}
Shaowei Wang, Zhi-Hua Zhou, Mengyao Ge, and Chonggang Wang,
\newblock ``Resource allocation for heterogeneous cognitive radio networks with
  imperfect spectrum sensing,''
\newblock {\em IEEE Journal on Selected Areas in Communications}, vol. 31, no.
  3, pp. 464--475, March 2013.

\bibitem{shakkottai2002scheduling}
Sanjay Shakkottai and Rayadurgam Srikant,
\newblock ``Scheduling real-time traffic with deadlines over a wireless
  channel,''
\newblock {\em Wireless Networks}, vol. 8, no. 1, pp. 13--26, 2002.

\bibitem{kang2013performance}
X.~Kang, W.~Wang, J.J. Jaramillo, and L.~Ying,
\newblock ``On the performance of largest-deficit-first for scheduling
  real-time traffic in wireless networks,''
\newblock in {\em Proceedings of the fourteenth ACM international symposium on
  Mobile ad hoc networking and computing}. ACM, 2013, pp. 99--108.

\bibitem{neely2003power}
Michael~J Neely, Eytan Modiano, and Charles~E Rohrs,
\newblock ``Power allocation and routing in multibeam satellites with
  time-varying channels,''
\newblock {\em IEEE/ACM Transactions on Networking}, vol. 11, no. 1, pp.
  138--152, 2003.

\bibitem{Two_Q_Light_Hvy}
K.~Jagannathan, M.G. Markakis, E.~Modiano, and J.N. Tsitsiklis,
\newblock ``Throughput optimal scheduling over time-varying channels in the
  presence of heavy-tailed traffic,''
\newblock {\em Information Theory, IEEE Transactions on}, vol. 60, no. 5, pp.
  2896--2909, May 2014.

\bibitem{li2011delay}
C.-P. {Li} and M.J. {Neely},
\newblock ``{Delay and Power-Optimal Control in Multi-Class Queueing
  Systems},''
\newblock {\em ArXiv e-prints}, Jan. 2011.

\bibitem{ewaisha2015dynamic}
Ahmed Ewaisha, Cihan Tepedelenlio, et~al.,
\newblock ``Dynamic scheduling for delay guarantees for heterogeneous cognitive
  radio users,''
\newblock in {\em 2015 49th Asilomar Conference on Signals, Systems and
  Computers}. IEEE, 2015, pp. 169--173.

\bibitem{Ewaisha_Globecom16}
Ahmed Ewaisha and Cihan Tepedelenlioglu,
\newblock ``Delay optimal joint {Scheduling-and-Power-Control} for cognitive
  radio uplinks,''
\newblock in {\em 2016 IEEE Global Communications Conference: Cognitive Radio
  and Networks (Globecom'16 - CRN)}, Washington, USA, Dec. 2016.

\bibitem{haykin2005cognitive}
Simon Haykin,
\newblock ``Cognitive radio: brain-empowered wireless communications,''
\newblock {\em Selected Areas in Communications, IEEE Journal on}, vol. 23, no.
  2, pp. 201--220, 2005.

\bibitem{Bari13ciss}
M.~Bari, H.~Mustafa, and M.~Doroslova\v{c}ki,
\newblock ``Performance of the instantaneous frequency based classifier
  distinguishing {BFSK} from {QAM} and {PSK} modulations for asynchronous
  sampling and slow and fast fading,''
\newblock in {\em Proc. 47th Conference on Information Sciences and Systems},
  Johns Hopkins University, Baltimore, MD, Mar. 20-22 2013.

\bibitem{Bari13asilomar}
M.~Bari and M.~Doroslova\v{c}ki,
\newblock ``Quickness of the instantaneous frequency based classifier
  distinguishing {BFSK} from {QAM} and {PSK} modulations,''
\newblock in {\em Proc. 47th Annual Asilomar Conference on Signals, Systems,
  and Computers}, Pacific Grove, CA, USA, Nov. 3-6 2013, pp. 836--840.

\bibitem{Bari14asilomar}
M.~Bari and M.~Doroslova\v{c}ki,
\newblock ``Distinguishing {BFSK} from {QAM} and {PSK} by sampling once per
  symbol,''
\newblock in {\em Proc. 48th Annual Asilomar Conference on Signals, Systems,
  and Computers}, Pacific Grove, CA, USA, Nov. 2-5 2014.

\bibitem{Bari15spl}
M.~Bari and M.~Doroslova\v{c}ki,
\newblock ``Simple features for separating {CPFSK} from {QAM} and {PSK}
  modulations,''
\newblock {\em IEEE Signal Processing Letters}, vol. 22, no. 5, pp. 613--617,
  May 2015.

\bibitem{Bari15ciss1}
M.~Bari and M.~Doroslova\v{c}ki,
\newblock ``Robust recognition of linear and nonlinear digital modulations of
  {RRC} pulse trains,''
\newblock in {\em Proc. 49th Conference on Information Sciences and Systems},
  Johns Hopkins University, Baltimore, MD, Mar. 18-20 2015.

\bibitem{Bari15asilomar1}
M.~Bari and M.~Doroslova\v{c}ki,
\newblock ``Separation of signals consisting of amplitude and instantaneous
  frequency {RRC} pulses using {SNR} uniform training,''
\newblock in {\em Proc. 49th Annual Asilomar Conference on Signals, Systems,
  and Computers}, Pacific Grove, CA, USA, Nov. 8-11 2015.

\bibitem{Bertsekas_Data_Networks}
D.i Bertsekas and R.~Gallager,
\newblock {\em Data Networks (2Nd Ed.)},
\newblock Prentice-Hall, Inc., Upper Saddle River, NJ, USA, 1992.

\bibitem{semiconductor2008long}
Freescale Semiconductor,
\newblock ``Long term evolution protocol overview,''
\newblock {\em White Paper, Document No. LTEPTCLOVWWP, Rev 0 Oct}, 2008.

\bibitem{Neely_CNC_2009}
R.~Urgaonkar and M.J. Neely,
\newblock ``Opportunistic scheduling with reliability guarantees in cognitive
  radio networks,''
\newblock {\em Mobile Computing, IEEE Transactions on}, vol. 8, no. 6, pp.
  766--777, 2009.

\end{thebibliography}

\end{document}